\documentclass{llncs}[10point]

\usepackage[lncs,amsthm,autoqed]{pamas}
\usepackage{txfonts}
\usepackage{arydshln}
\usepackage{wrapfig}

\usepackage{enumitem,lipsum}

\newenvironment{tablenotes}{%
	\begin{list}{\labelitemi}{%
		\setlength{\leftmargin}{2em}%
		\setlength{\itemsep}{0pt}%
		}}{%
		\end{list}\vspace*{-.8\baselineskip}}

\newcommand{\set}[1]{\{#1\}}

\newcommand\eat[1]{}

\usepackage{tikz}
\usetikzlibrary{arrows}

\usepackage{mathrsfs}
\usepackage{slashbox}

\newcommand{\nash}{\text{{NS}}\xspace}
\newcommand{\indiv}{\text{{IS }}\xspace}
\newcommand{\contract}{\text{CIS}\xspace}

\usepackage{calrsfs}

\newcommand{\pref}{\succsim\xspace}
\newcommand{\Pref}[1][]{
	\ifthenelse{\equal{#1}{}}{\mathrel \succsim}{\mathop{\succsim_{#1}}}
}                                          
\newcommand{\sPref}[1][]{                  
	\ifthenelse{\equal{#1}{}}{\mathrel \succ}{\mathop{\succ_{#1}}}
}                                          
\newcommand{\Indiff}[1][]{                 
	\ifthenelse{\equal{#1}{}}{\mathrel \sim}{\mathop{\sim_{#1}}}
}
\newcommand{\prefset}[1][]{\ifthenelse{\equal{#1}{}}{\mathcal{\succsim}}{\mathcal{\succsim}_{#1}}}

\newcommand{\midd}{\makebox[2ex]{$:$}}

\newcommand{\allelse}[1][3em]{\;\rule[.5ex]{#1}{.4pt}\;}

\newlength{\wordlength}

\title{Stable marriage and roommate problems with individual-based stability\thanks{This material is based upon work supported by the Deutsche Forschungsgemeinschaft under grant BR 2312/10-1
}}
\author{Haris Aziz}
\institute{%
  Department of Informatics,
  Technische Universit\"at M\"unchen \\
  85748 Garching bei M\"unchen, Germany \\
  \email{aziz@in.tum.de}
  }

\begin{document}

\maketitle

\begin{abstract}
	Research regarding the stable marriage and roommate problem has a long and distinguished history in mathematics, computer science and economics. Stability in this context is predominantly core stability or one of its variants in which each deviation is by a group of players. There has been little focus in matching theory on stability concepts such as Nash stability and individual stability in which the deviation is by a single player. Such stability concepts are suitable especially when trust for the other party is limited, complex coordination is not feasible, or when only unmatched agents can be approached. Furthermore, weaker stability notions such as individual stability may in principle circumvent the negative existence and computational complexity results in matching theory. 
	 We characterize the computational complexity of checking the existence and computing individual-based stable matchings for the marriage and roommate settings. 
	 One of our key computational results for the stable marriage setting also carries over to different classes of hedonic games for which individual-based stability has already been of much interest.
\end{abstract}

\section{Introduction}

In stable matching problems, the aim is to match agents in a stable manner to objects or to other agents, keeping in view the preference of the agents involved. These problems have significant applications in matching residents to hospitals, students to schools, etc. and have received tremendous interest in the mathematical economics, computer science and operations research communities~\citep[see \eg ][]{GuIr89a,RoSo90a}. Informally, a matching is deemed `stable' if the agents do not have an incentive to deviate to achieve a better matching for themselves. 
In the matching theory literature, the predominant notion of stability is indeed the core in which no pair of agents prefer to be matched to each other than remain in their current matching. \emph{Core stability} (also simply called stability) has been extensively investigated in the context of the \emph{stable marriage (SM)} problem~\citep{GaSh62a} and \emph{stable roommate (SR)} problem~\citep{Irvi85a} which are two of the most fundamental settings in matching theory.  
A comprehensive survey of the stable marriage and roommate problems is present in \citep{GuIr89a}.

We formulate the stable marriage and stable roommate settings as \emph{marriage games} and \emph{roommate games}. Both of these games are basic subclasses of \emph{hedonic coalition formation games} in which an agent's preference of a partition only depends on the coalition (of arbitrary size) he is a member of and not on how the remaining agents are grouped~\citep[see \eg ][]{BoJa02a,Hajd06a}.  Of course, in the roommate and marriage games, feasible partitions simply correspond to matchings because each coalition is of size at most two. 
The main focus in hedonic games has been on different natural notions of stability of partitions. 
 The stability concepts include individual-based stability concepts (\emph{Nash stability (NS)}, \emph{individual stability (IS)}, and \emph{contractual individual stability (CIS)}) and group-based stability concepts (\emph{core (C)} and \emph{strict core (SC)})~\citep[see \eg][]{BoJa02a}. Another individual-based stability concept is \emph{contractual Nash stability (CNS)} which is stronger than CIS and is defined in an analogous way to IS~\citep{Sudi07b}.

In this paper, we characterize the complexity of checking existence of and computing individual-based stable outcomes in marriage and roommate games. A number of existence results are also presented. Our results shed further light on the dynamics of stability concepts like Nash stability in fundamental settings such as marriage games. 

There are a number of reasons why individual-based stability in matching models may be of interest.
Individual-based stability applies in situations when forming arbitrary new coalitions may be `\emph{costly or may require complex coordination among the players}'\citep{Papa07b}. Furthermore, `\emph{if information on the preferences of other players is scarce [$\ldots$], then considering the actions of individual players only may be quite compelling}'\citep{Papa07b}. 
Since marriage games and roommate games may not admit a strict core stable and core stable matching respectively~\citep{GuIr89a}, it makes sense to examine weaker stability notions such as IS. IS may also apply to other matching models. For e.g., in hospital-resident matching, the hospital may not deviate with a resident and may accept any acceptable candidate (with a minimum level of competency). 

Since marriage and roommate games are two of the most fundamental classes of hedonic games, our results have bearing on the coalition formation literature.  
In fact, one of our key computational results for marriage games also carries over to different classes of hedonic games for which individual-based stability has already been of much interest. 
Finally, we point out that marriage and roommate settings are also one of the most basic models in network formation which further motivates our study.

\section{Related literature}
	The complexity of computing partitions which are Nash stable or individual stable has previously been examined for some classes of hedonic games such as additively separable hedonic games~\citep[see \eg ][]{Olse09a,SuDi10a}, and hedonic games represented by individually rational lists of coalitions (RIRLC)~\citep{Ball04a}. \citet{Sudi07b} introduced CNS and showed that a CNS partition is guaranteed for separable hedonic games satisfying weak mutuality. 
\citet{Papa07b} used restrictions on acceptable coalitions to characterize classes of hedonic coalition formation games with strict preference for which Nash stable and individual stable partitions are guaranteed to exist. \footnote{As a result, \citet{Papa07b} proved that for marriage games \emph{with strict preferences}, an individually stable partition is guaranteed to exist. The proof is non-constructive and an algorithm to compute an IS matching was not presented.}

		For the roommate and marriage settings, there has been considerable work on the \emph{stable marriage (SM)} problem and \emph{stable roommate (SR)} problem. Stability in this regard is mostly \emph{core stability} (also simply called stability). In the stable marriage (SM) problem , the set of agents is partitioned into men and women; men and women express strict preferences over all their counterparts; and the aim is to find a stable matching in which men and women are matched to each other. The stable roommate problem (SR) is the unisex generalization of the stable roommate problem in which roommates are paired with each other in a stable matching~\citep{Irvi85a}.

	Subsequently, variants of the problems SM and SR have been examined: i) SMI and RMI --- stable marriage and stable roommate problems with \emph{incomplete} preference preference lists thereby signifying that the agents not in a preference list of an agent are unacceptable to the agent; ii) SMT and RMT --- cases which allow ties/indifferences in the preference; iii) and finally SMTI and RMTI --- cases which allow both ties and incomplete lists. The reader may refer to Table~\ref{table:rg-mg-complexity3} which summarizes the complexity results in the literature concerning stable matchings where the stability concept concerns deviation by groups or pairs of players. We will cover all the stable marriage and stable roommate settings mentioned above but instead of considering core stability, we will consider individual-based stability concepts.
	
	Recently, \citet{AGM+11a} and \citet{NSV+11a} modeled uncoordinated marriage games and roommate games with strict preferences via a normal form game in which the pure Nash equilibria of the normal form game coincide with core stable matchings. 

		\begin{table*}[tbh]
		\centering
			\scalebox{0.9}{
		\centering

		\begin{tabular}{lllll}

		\toprule

		Game&Preference&Problem&Stability&Complexity\\
		&Restrictions&setting&concept&\\

		\midrule
		
		
	RG&No&	(SRT)&C&NP-C~\citep{Ronn90a}\\
&unacceptability&&&\\

		\midrule

		RG&Strict&(SRI)&C&in P~\citep{Irvi85a}\\
		\midrule
		RG&&(SRTI)&SC&in P~\citep{Scot05a}\\
		\midrule

		MG&&(SMTI)&C&in P~\citep{GaSh62a}$^a$\\
		MG&&(SMTI)&SC&in P~\citep{Irvi94a}\\

		\bottomrule
		\end{tabular}
		}
		\begin{tablenotes}
			\item[a] 
			Only combination of game setting and stability concept in the table for which a stable matching is guaranteed to exist.		
		\end{tablenotes}
		\caption{Complexity of computing a stable matching in marriage and roommate games for group-based stability: literature summary. In roommate and marriage games, core stability and pairwise-stability coincides.}
		\label{table:rg-mg-complexity3}

		\end{table*}
		

		\vspace{-2.5em}

\section{Preliminaries}

\paragraph{Hedonic games.}
We review the terminology and notation used in this paper.
Let~$N$ be a set of~$n$ players. A \emph{coalition} is any non-empty subset of~$N$. By $\mathcal{N}_i$ we denote the set of all coalitions player~$i$ may belong to, \ie $\mathcal N_i=\set{S\subseteq N\midd i\in S}$. A \emph{coalition structure}, or simply a \emph{partition}, is a partition~$\pi$ of the players $N$ into coalitions, where $\pi(i)$ is the coalition player~$i$ belongs to. 

A \emph{hedonic game} is a pair~$(N,\pref)$, where $\pref=(\pref_1,\dots,\pref_n)$ is a \emph{preference profile} specifying the preferences of each player~$i$ as a binary, complete, reflexive, and transitive \emph{preference relation~$\Pref[i]$} over~$\mathcal N_i$. If~$\pref_i$ is also anti-symmetric we say that~$i$'s preferences are \emph{strict}. Note that $S\sPref[i]T$ if $S\Pref[i]T$ but not~$T\Pref[i]S$---\ie if~$i$ \emph{strictly prefers} $S$ to~$T$---and $S\Indiff[i]T$ if both $S\Pref[i]T$ and $T\Pref[i]S$---\ie if~$i$ is \emph{indifferent} between~$S$ and~$T$.

For a player~$i$, a coalition~$S$ in~$\mathcal N_i$ is \emph{acceptable} if for~$i$ being in~$S$ is at least as preferable as being alone---\ie if $S\Pref[i]{\set i}$---and \emph{unacceptable} otherwise. If $\{i,j\} \sPref_i \{i\}$, then we say that $i$ \emph{likes} $j$. 
We also say that partition~$\pi$ is \emph{acceptable} or \emph{unacceptable} to a player~$i$ according to whether~$\pi(i)$ is acceptable or unacceptable to~$i$, respectively. Moreover,~$\pi$ is \emph{individually rational (IR)} if~$\pi$ is acceptable to all players.

\paragraph{Roommate \& marriage games.}
A \emph{roommate game (RG)} is a hedonic game $(N,\pref)$ in which for each $i\in N$, coalitions of size three or more are unacceptable and preferences $\pref$ over other players are extended naturally over preferences over coalitions in the following way: 
$\{i\} \cup \{j\} \succsim_i \{i\} \cup \{k\}$ if and only if $j\succsim_i k$ for all $i,j,k\in N$.
In the matching theory literature, preferences $\pref_i$ of player $i$ over other players are represented via preferences list so that if $j\neq i$ is not on the preference list of $i$, then $j$ is unacceptable to $i$. A \emph{marriage game (MG)} is a roommate game in which $N$ is partitioned into two sets $M$ (men) and $W$ (women) such that each agent considers a member of his own sex unacceptable. By marriage games with no unacceptability, we will mean preferences such that each player considers a member of the opposite sex acceptable.\footnote{Such a setting is also referred to as marriage problem with complete lists.} Similarly, by roommate games with no unacceptability, we mean that each player consider all other players acceptable. When we refer to a matching, we will mean the obvious partition in which the unmatched players are in singleton coalitions.

\paragraph{Stability Concepts.}

We now present the standard stability concepts for hedonic games. 
The following are standard stability concepts based on deviations by individual players~\citep[see \eg][]{BoJa02a,Sudi07b,GaSa10a,GaSa11a,ABS11b,Olse09a, SuDi10a}.

	\begin{itemize}
\item A partition is  \emph{Nash stable (NS)} if no player can benefit by 
	moving from his coalition to another (possibly empty) coalition $T$.
\item A partition is \emph{individually stable (IS)} if no player can
	benefit by moving from his coalition to another existing (possibly empty) coalition $T$  while not making the members of $T$ worse off. 
	\item A partition is  \emph{contractually individually stable (CIS)} if no
	player can benefit by moving from his coalition $S$ to another existing (possibly empty) coalition
	$T$ while making neither the members of $S$ nor the members of
	$T$ worse off. 
	\item 
	A partition is \emph{contractual Nash stable (CNS)} if no player can
	benefit by moving from his coalition $S$ to another existing (possibly empty) coalition $T$  while not making the members of $S$ worse off.\footnote{A suitable example if that of a criminal organization where joining may be easy but moving out requires permission from the other members. In the frivolous marriage parlance, this can be interpreted as requiring permission for divorce.}
	
	\end{itemize} 
	
	Depending on the context, we will utilize abbreviations NS, IS, CNS, and IR etc. either for adjectives (for e.g. IS for individually stable) or for nouns (for e.g. IS for individual stability). 
	\emph{Core (C) } and \emph{strict core (SC)} are defined via deviations by a group or pair of players.
	We refer to \citep{BoJa02a} for the precise definitions. In the restricted domain of roommate and marriage games, core stability and strict core stability correspond respectively to pairwise stability and strong pairwise stability.  

Based on their definitions, we see how stability concepts are related to each other. 
The inclusion relationships between stability concepts depicted in Figure~1 follow from the definitions of the concepts.

\vspace{-0.5em}

\begin{figure}
	    \centering
	    \scalebox{0.9}{
		\begin{tikzpicture}
			\tikzstyle{pfeil}=[->,>=angle 60, shorten >=1pt,draw]
			\tikzstyle{onlytext}=[]

			\node[onlytext] (NS) at (1,0) {NS};
			\node[onlytext] (SC) at (3,0) {SC};
			\node[onlytext] (IS) at (2,-1.5) {IS};
			\node[onlytext] (CNS) at (0,-1.5) {CNS};

			\node[onlytext] (C) at (4,-1.5) {C};
			\node[onlytext] (CIS) at (1,-3) {CIS};
			\node[onlytext] (IR) at (3,-3) {IR};

			\draw[pfeil] (NS) to (IS);
			\draw[pfeil] (SC) to (IS);
			\draw[pfeil] (IS) to (CIS);
			\draw[pfeil] (SC) to (C);

			\draw[pfeil] (C) to (IR);
			\draw[pfeil] (IS) to (IR);
			\draw[pfeil] (NS) to (CNS);
			\draw[pfeil] (CNS) to (CIS);

		\end{tikzpicture}
		}
		\caption{Inclusion relationships between stability concepts. E.g., every NS partition is also IS.}\label{fig:tnfigure}
	  \end{figure}
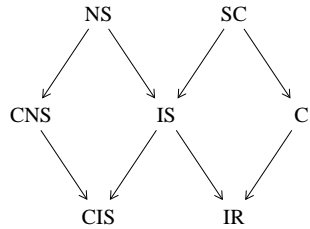

\vspace{-1.5em}
	\section{Negative results} \label{sec:neg}

We first start with some bad news which contrasts sharply with the fact that a core stable matching can be computed efficiently for marriage games. 

	\begin{theorem}\label{th:mg-ns-npc}
		For marriage games, checking whether there exists a NS matching is NP-complete.
	\end{theorem}
	\begin{proof}[Sketch]
We present a polynomial-time reduction from {\sc MinimumMaximalMatching} (MMM) to checking whether there exists a NS matching for a marriage game.\\

	\noindent
		Name: {\sc MinimumMaximalMatching} (MMM).\\ 
		Instance: Graph $G=(V,E)$ and integer $k\in Z^+$. \\
		Question: Does G have a maximal matching M with size $\leq k$?	\\

		{\sc MinimumMaximalMatching} (MMM) is NP-complete even for subdivision graphs~\citep{HoKi93a,MII+02a}. 
		Let $G=(V,E)$ and and integer $k\in Z^+$ be the instance of MMM. Graph $G$ is the subdivision graph of some graph $G'=(V',E')$ such that $V=V'\cup E'$ and $E=\{\{e,v\}\midd e\in E', \textrm{~and~} v \in V' \textrm{~and~} v~~\textrm{is incident to $e$ in $G'$}\}$. It is easy to see that $G$ is a bipartite graph where $V=A \cup B$. We may assume that $|A|=|B|=n$. \footnote{If this were not the case and $|A|=|B|+r$, then we can add $r$ vertices $a_1, \ldots, a_r$ to $A$ and $2r$ vertices to $b_1,\ldots b_r$, $c_1, \ldots, c_r$ to $B$ where $a_i$ is adjacent to $b_i$ and $c_i$ for each $i$ ($1\leq i\leq r$). Then $G$ has a maximal matching of size $k$ if and only if the reduced graph has a maximal matching of size $k+r$.} 

		We construct a marriage game $(N,\pref)$ where $N=V\cup X \cup\{y\}$ where $X=\{x_1,\ldots, x_{n-k}\}$ and the player preferences are as follows where players in the same set in the preference list are equally preferred:
		\begin{align*}
	a 	&: \{b'\in B\midd \{a,b'\}\in E\} \succ_a X \succ_a a \succ_a \allelse 	& \forall a\in A\\
	b   &:\{a'\in A \midd \{a',b\} \in E\} \succ_b b \succ_b \allelse	&\forall b\in B\\
	x   &:A \sim y \succ_x x \succ_x \allelse	&\forall x\in X\\
	y  &:y \succ_y \allelse	&
	\end{align*}
		Then, the claim is that there exists a maximal matching of size at most $k$ if and only if there exists a Nash stable matching.

	If there exists a maximal matching $M$ of size at most $k$ of $(V,E)$, there are at most $k$ players in $A$ which are matched. This implies that there are least $n-k$ players in $A$ which are not matched in $M$. We match $n-k$ players from $A$ which are not in matched in $M$ with players in $X$. Therefore all the players in $X$ are perfectly matched to players in $A$. Since each player in $X$ is matched to a player in $A$ and has no incentive to deviate to $y$. Each player in $A$ is either matched to a player in $B$ or $X$ or is unmatched. If $a\in A$ is matched to an acceptable player $b\in B$, it is perfectly happy and has no incentive to deviate. If $a\in A$ is not matched to some acceptable player in $B$, it is either in a singleton coalition or is paired with a player in $X$. In either case, $a$ does not have an incentive to deviate elsewhere because all players in $X$ are perfectly matched, and there is no unmatched $b\in B$ such that $(a,b)\in E$. If this were the case, then $M$ was not a maximal matching. Similarly, each unmatched $b\in B$ has no incentive to deviate to an unmatched $a\in A$ because of maximality of $M$.	

	Assume that each maximal matching in $G$ has size greater than $k$. Then for any such matching $M'$, there are at most $n-k-1$ unmatched players in $A$ which need to be matched with $n-k$ players in $X$. Therefore one player in $X$ will not be matched to a player in $A$ and will try to match with $y$ who wants to be alone. Therefore there exists no NS matching. It could have been the case that if $M'$ was not maximal, there would be enough free players from $A$ which could cater for players in $X$. Consider a matching $M''$ of size less than or equal to $k$ which is not maximal. Then, there are at least $n-k$ players in $A$ which can form a perfect matching with players in $X$. However since $M''$ is not maximal, some player $a\in A$ which is matched to an $x\in X$ wants to deviate to an unmatched acceptable player in $B$. 
	\end{proof}
	

As a corollary of Theorem~\ref{th:mg-ns-npc}, we obtain corresponding results for a number of hedonic games including \emph{hedonic games in RIRLC} ~\citep{Ball04a}, \emph{additively separable hedonic games}~\citep{Olse09a, SuDi10a} and \emph{hedonic games based on the best or worst players}~\citep{CeRo01a,AHP12a}.  
Before, we present the corollary, we will define the other classes of hedonic games for the help of the reader.


In the \emph{Representation by Individually Rational Lists of Coalitions} (RIRLC) for hedonic games, each player expresses his preferences only over his acceptable coalitions~\citep{Ball04a}. 
An \emph{additively separable hedonic game (ASHG)}  is pair $(N,v)$ such that each player $i\in N$ has value $v_i(j)$ for player $j$ being in the same coalition as $i$ and if $i$ is in coalition $S\in \mathcal{N}_i$, then $i$ gets utility $\sum_{j\in S\setminus \{i\}}v_i(j)$. For coalitions $S,T\in\mathcal{N}_i$, $S \succsim_i T$ if and only if $\sum_{j\in S\setminus \{i\}}v_i(j) \geq \sum_{j\in T\setminus \{i\}}v_i(j)$. 
Finally, we define B-hedonic games and W-hedonic games based on the best and worst players.
For a subset~$J$ of players, we denote by $\max_{\pref_i}(J)$ and $\min_{\pref_i}(J)$ the sets of the most and least preferred players in~$J$ by~$i$, respectively. We will assume that $\max_{\pref_i}(\emptyset)=\min_{\pref_i}(\emptyset)=\{i\}$. 
In a \emph{B-hedonic game} the preferences~$\pref_i$ of a player~$i$ over players extend to preferences over coalitions in such a way that, for all coalitions~$S$ and~$T$ in~$\mathcal N_i$, we have $S\Pref[i]T$ if and only if  
either some player in~$T$ is unacceptable to~$i$ or all players in~$S$ are acceptable to~$i$ and 
$j\mathop{\pref_i}k$ for all $j\in\max_{\pref_i}(S\setminus \{i\})$ and $k\in\max_{\pref_i}(T\setminus \{i\})$. Analogously, in a \emph{W-hedonic game $(N,\pref)$}, we have $S\Pref[i]T$ if and only if either some player in~$T$ is unacceptable to~$i$ or
$j\mathop{\pref_i}k$ for all $j\in\min_{\pref_i}(S\setminus \{i\})$ and $k\in\min_{\pref_i}(T\setminus \{i\})$.


\begin{corollary}\label{cor:ns-npc}
The problem of checking whether there exists a NS matching is NP-complete for the following: i) roommate games, ii) hedonic games in RIRLC~\citep{Ball04a}, and iii) additively separable hedonic games~\citep{Olse09a, SuDi10a}. 
iv) B-hedonic games, and v) W-hedonic games.
\end{corollary}
\begin{proof}
We address each of the cases separately.
\renewcommand{\labelenumi}{\roman{enumi}.}
\noindent
\begin{enumerate}
\item Roommate games are a generalization of marriage games with as compact a representation.
\item A marriage game can be reduced in linear time to a hedonic game in RIRLC which is linear in the size of the marriage game. Instead of each player having preferences over players, it has preferences over coalitions of size two with each coalition of course including the player himself. 
\item A marriage game $(N,\pref)$ can be reduced to an ASHG $(N,v)$ in which $v$ is defined as follows: $v_i(i)=0$; $v_i(j)\geq v_i(k)$ if and only if $j\succsim_i k$; and $v_i(j)$ is a suitably large negative valuation if $j$ is unacceptable to $i$ in $(N,\pref)$. Then, in any IR partition $\pi$ in game $(N,v)$, 
the pigeon-hole principle ensures that two members of the same sex are never together in the same coalition. It follows that an IR partition $\pi$ for $(N,v)$ is an individually rational matching for marriage game $(N,\pref)$. Furthermore, $\pi$ is NS in ASHG $(N,v)$ if and only if it is NS in the marriage game $(N,\pref)$.
\item  A marriage game $(N,\pref)$ is a B-hedonic game $G$ in which acceptable coalitions are of size one or two. If a coalition were of size more than two, then members of the same sex would be present which makes the coalition unacceptable. Therefore a partition is NS for $G$ if and only if it is NS in the marriage game $(N,\pref)$.
\item Same argument for W-hedonic games as for B-hedonic games.

\end{enumerate}
This complete the proof.
\end{proof}
	
Some of the statements in Corollary~\ref{cor:ns-npc} were known but required separate proofs for each particular class~~\citep{Ball04a, Olse09a, SuDi10a, AHP12a}.


It can also be shown that for roommate games, even checking whether there exists an IS matching is NP-complete. The proof utilizes a three-player roommate game for which no IS partition exists.

	\begin{theorem}\label{th:rg-is-npc}
		For roommate games, checking whether there exists an IS matching is NP-complete.
	\end{theorem}
		\begin{proof}[Sketch]
			We present a polynomial-time reduction from {\sc MinimumMaximalMatching} (MMM) to checking whether there exists an IS matching for a roommate game. 
			
			As in the proof of Theorem~\ref{th:mg-ns-npc}, we assume that the instance of MMM is a bipartite graph in which the vertices $V$ are partitioned into sets $A$ and $B$. 
			We construct a roommate game $(N,\pref)$ where $N=V\cup X$ where $X=\{x_1^0, x_1^1, x_1^2,\ldots, x_{n-k}^0, x_{n-k}^1, x_{n-k}^2\}$ and the player preferences are as follows:
	\begin{align*}
	a 	&: \{b\in B\midd \{a,b'\}\in E\} \succ_a X \succ_a a \succ_a \allelse& \forall a\in A\\
	b   &: \{a'\in A \midd \{a',b\} \in E\} \succ_b b \succ_b \allelse	&\forall b\in B\\
	x_i^{j}&:A \sim_{x_i^{j}} x_i^{(j+1) mod 3} \succ_{x_i^j} x_i^{(j-1) mod 3} \succ_{x_i^j} x_i^j \succ_{x_i^j}\allelse	&\\
	&	\forall i\in \{1,\ldots,n-k \} , j\in \{0,1,2\}
	\end{align*}
	The preferences are set in such a way that each triplet $x_i^0, x_i^1, x_i^2$ for $i\in \{1,\ldots n-k\}$ is in a perpetual cycle of deviations $\{x_i^0, \{x_i^1, x_i^2\}\}$, $\{\{x_i^2,x_i^0\}, \{x_i^1\}\}$, $\{\{x_i^2\}, \{x_i^0,x_i^1\}\}$, and back to  $\{x_i^0, \{x_i^1, x_i^2\}\}$. And this cycling can only be stopped by the help of a players in $A$ which is not already in a coalition with a player $b\in b$ such that $(a,b)\in E$. 

	The claim is that MMM has a `yes' instance (i.e., there exists a maximal matching of size $k$) if and only if there exists an IS  matching for $(N,\pref)$. If there exists a maximal matching $M$ of size $k$ or less, then there are at least $n-k$ players in $A$ which are not matched in $M$. Then for each $i\in \{1,\ldots, n-k\}$ a player $x_i^0$ is matched to one of the unmatched $a\in A$. The $x_i^0$'s have no incentive to deviate any where. Each $x_i^1$ and $x_i^2$ are matched to each other and have no incentive to deviate. Similarly, no player in $B$, $A$, and $\{y\}$ has an incentive to deviate. Thus there exists an IS matching.

		Assume that each maximal matching $M'$ in $G$ has size greater than $k$. Then, there are at most $n-k-1$ players in $A$ which are not matched in $M'$ and which can be matched with $n-k$ players in $X$. Therefore not all $n-k$ triplet IS cycles of $x_i^0, x_i^1, x_i^2$ can be disrupted. Therefore, there does exist one set of players $x_i^0, x_i^1, x_i^2$ which is in a perpetual IS cycle. Therefore there exists no IS matching. It could have been the case that if $M'$ was not maximal, there would be enough free players from $A$ which could cater for players in $X$. Consider a matching $M''$ of size less than or equal to $k$ which is not maximal. Then, there are at least $n-k$ players in $A$ which can form a perfect matching with players in $X$. However since $M''$ not maximal, some player in $a$ wants to deviate to some $b\in B$ such that $(a,b)\in E$.
		\end{proof}


\section{Positive results}

In this section, we present a number of positive computational and existence results concerning individual-based stability in marriage and roommate games.
Firstly, we can use a potential-function argument to show the following.
	
\begin{proposition}\label{prop:RG-CIS-easy}
For every roommate game, a CIS and IR matching exists and can be computed in $O(n^2)$. 
\end{proposition}
\begin{proof}[Sketch]
Take the IR partition of singletons. 
If the partition is CIS, we are done. 
Otherwise, if there is a feasible CIS deviation, enable it. 
In each CIS deviation at least one player strictly improves his utility and no player's utility decreases. 
Since there can only be a maximum of $n(n-1)$ CIS deviations, a CIS and IR partition is obtained in $O(n^2)$. 
\end{proof}

What is much more surprising is that although an IS matching is not guaranteed to exist for roommate games (Theorem~\ref{th:rg-is-npc}), a CNS matching is.\footnote{Therefore, we identify an important class of hedonic games other than weakly mutual separable hedonic games~\citep{Sudi07b} for which a CNS partition is guaranteed.}

\begin{theorem}\label{th:ops}
For every roommate game, a CNS matching exists and can be computed in $O(n^2)$. 
\end{theorem}
\begin{proof}[Sketch]
Let $\pi$ be the partition of singletons.
Set $B$, the set of players with at least one CNS deviation to the empty set. 
Let arbitrary CNS deviations take place from $\pi$ and update $B$ accordingly. 
The argument is that CNS deviations will not cycle, at least not if the starting configuration is $\pi$.  
Let $i\in N$ be the player which deviates from his current coalition $\pi(i)$ to another coalition $\{j\}$. Then, $i$ is guaranteed to never decrease his utility because he does not permit $j$ to move 
away. Clearly no player $k$ wants to join $\{i,j\}$ as a coalition of size three is unacceptable to each $k$. Therefore, players in $B$ can only improve their utility and cannot end up in a previous partition. 
In each deviation, either set $B$ grows or a player in $B$ improves his utility. Therefore, there can only be $O(n^2)$ deviations until a CNS matching is achieved.
\end{proof}

Interestingly, even in the absence of any preference restrictions, marriage games admit at least one IS matching which can be computed efficiently.

\begin{theorem}\label{th:mg-is-oddcycle}
For marriage games, an IS matching is guaranteed to exist.
Moreover, it can be computed in $O(n^2)$.
\end{theorem}
\begin{proof}[Sketch]
We present a constructive proof of the existence of an individually stable matching for marriage games. It is already known that a (core) stable matching exists for marriage games. However, in the presence of ties, core stability does not imply individual stability.

Given a marriage game $(N,\pref)$, first \emph{raise} the preferences $\pref$ to $\pref'$ in order to obtain the new modified game $(N,\pref')$. By raising preferences from $\pref$ to $\pref'$, we mean the following: for all $b,c\in N\setminus \{a\}$, $b \pref_{a} c$ if and only if $b \pref_{a}' c$ but if $b \sim_{a} a$ then $b \succ_{a}' a$. Now run the women-optimal version of the Gale-Shapley algorithm on $(N,\pref')$.
The claim is that the resultant matching $\pi$ is IS for the original game $(N,\pref)$. Note that there are no core deviations in $\pi$ according to preferences $\pref$ and also $\pref'$.

Since core stability implies individual rationality, the only possible IS deviations according to preferences $\pref$ are as follows:
1. Matched man has a valid IS deviation to an unmatched woman;  
2. Matched woman has a valid IS deviation to an unmatched man;
3. Unmatched woman has a valid IS deviation to an unmatched man;
4. Unmatched man has a a valid deviation  to an unmatched woman;

\begin{enumerate}
	\item If man $m_i$ matched to $w_j$ has a valid IS deviation to an unmatched woman $w_k$, then clearly $w_k \succ_{m_i} w_j$ and $m_i \sim_{w_k} w_k$. This means that $m_i \succ_{w_k}' w_k$ and woman $w_k$ must have proposed to $m_i$ in the algorithm. If $m_i$ was single, he would have paired with $w_k$ and if he was engaged to $w_j$, then he would broken off the engagement and paired with $w_k$. 
	\item If woman $w_j$ matched to $m_i$ has a valid IS deviation to an unmatched man $m_l$, then clearly 
	$m_l \succ_{w_j} m_i$ and $w_j \succsim_{m_l} m_l$. Therefore, $w_j$ would have proposed to $m_l$ before she proposed to $m_i$. Therefore $\{w_j,m_l\}$ would have been a pair in the first place.
	\item If unmatched women $w_j$ has a valid IS deviation to unmatched man $m_j$, then $m_i \succ_{w_j} w_j$ and woman $w_j$ would have proposed to $m_j$ in the algorithm and paired up with him.
	\item If unmatched man $m_i$ has a valid IS deviation to unmatched man $w_j$, then $w_j\succ_{m_i} m_i$ and $m_i\sim_{w_j} w_j$. This implies that $m_i\succ_{w_j}' w_j$. Therefore woman $w_j$ would have proposed to $m_i$ and they would have paired up. 	
\end{enumerate}
\end{proof}

\begin{table*}[t]
\centering
	\scalebox{0.9}{
\centering

\begin{tabular}{lllll}

\toprule

Game&Preference&Problem&Stability&Complexity\\
&Restrictions&setting&concept&\\

\midrule
RG&&(SRTI)&\contract \& IR&in P (Prop.~\ref{prop:RG-CIS-easy})\\
RG&&(SRTI)&\nash&NP-C~(Cor.\ref{cor:ns-npc})\\
RG&&(SRTI)&\indiv&NP-C~(Th.~\ref{th:rg-is-npc})\\
RG&&(SRTI)&CNS&in P (Th.~\ref{th:ops})\\


\midrule
RG&No unacceptability&(SRT)&\nash, \indiv&in P (Th.~\ref{th:rg-noaccept-is-easy})\\

\midrule

MG&&(SMTI)&\nash&NP-C~(Th.~\ref{th:mg-ns-npc})\\
MG&&(SMTI)&\indiv&in P (Th.~\ref{th:mg-is-oddcycle})\\

\midrule
MG&No unacceptability&(SMT)&\indiv&in P (Th.~\ref{th:mg-is-oddcycle})\\
MG&No unacceptability&(SMT)&\nash&in P~(Cor.~\ref{th:mg-ns-noaccept})\\

\bottomrule
\end{tabular}
}
\caption{Complexity of individual-based stability in marriage and roommate games. The NP-completeness result for checking the existence of a NS matching for a marriage game also applies to a number of representations and classes of hedonic games (Corollary~\ref{cor:ns-npc}).}
\label{table:rg-mg-complexity2}

\end{table*}

One may wonder whether the negative results in the previous section 
can be circumvented by disallowing players to express other players as unacceptable.
Preferences are \emph{mutual} if whenever a player considers another player acceptable, then the other player also considers the first player acceptable.

\begin{proposition}\label{th:rg-consistent-ns=is}
For roommate games and marriage games with mutual preferences, a NS matching exists if and only if an IS matching exists.
\end{proposition}
\begin{proof}
A NS matching is of course IS. Therefore the right implication follows trivially. 
Assume a matching $\pi$ is not NS. Then, there exist a pair $\{a_i,a_j\}$ in $\pi$ such that $a_i$ wants to deviate to a singleton coalition $\{a_k\}$. Since $a_k$ finds $a_i$ acceptable, therefore it does not object to $a_i$ joining him. This means that $\pi$ has a valid IS deviation and $\pi$ is not IS.
\end{proof}
\begin{corollary}\label{th:mg-ns-noaccept}
For marriage games with no unacceptability, a NS matching is guaranteed to exist.
Moreover, it can be computed in $O(n^2)$.
\end{corollary}

As a corollary, for marriage games with no unacceptability, a NS matching is guaranteed to exist.
Moreover, it can be computed in $O(n^2)$.
It was seen that there is a marked contrast between marriage and roommate games regarding the complexity of individual stability.  However, if there is no unacceptability, it can be checked efficiently whether a NS or IS matching exists for roommate games.

\begin{theorem}\label{th:rg-noaccept-is-easy}
For roommate games with no unacceptability, it can be checked in $O(n^4)$ whether a NS matching or an IS matching exists.
\end{theorem}
\begin{proof}[Sketch]
From Proposition~\ref{th:rg-consistent-ns=is}
, we just need to check whether a NS matching exists or not. 
If $n$ is even, then we are already done as any perfect matching of players is not only IR but also NS.
The problem becomes interesting if $n$ is odd. If there exists a NS matching in which there are more than one unmatched players (singleton coalitions), then there also exists a NS matching in which there is exactly one unmatched player. Therefore, we need to check whether there exists a NS matching with one unmatched player or not. Take any player $i$ to be the unmatched player. Then, we want to check whether the other players in $N\setminus \{i\}$ can be matched so that no player has an incentive to deviate to $i$. There is no other possible deviation as each player is matched with an acceptable player and does not have an incentive to become alone. Form an undirected graph $G_i=(V,E)$ such that $V=N\setminus \{i\}$ and $E$ is defined as follows: $\{j,k\}\in E$ if and only if $k\succsim_j i$ and $ j \succsim_k i$. Then, if there exists a perfect matching $M$ of $G_i$, return $M\cup \{\{i\}\}$ as the NS matching (it can be checked in  $O(n^3)$ whether a graph contains a perfect matching~\citep{Edmo65a}). It is clear that no players in $N\setminus \{i\}$ want to leave their partners and deviate to the unmatched $i$. 
Otherwise, repeat with another $i\in N$. If there exists no perfect matching in $G_i$ for all $i\in N$, return `no'.
\end{proof}

\vspace{-.25em}
\section{Conclusions}
\vspace{-.25em}

We examined the computation of individual-based stable outcomes in matching problems and focused on marriage and roommate problems. A complete characterization was achieved (please see Table~\ref{table:rg-mg-complexity2}). As a corollary of Theorem ~\ref{th:mg-ns-npc}, we also gave simple proofs for some results in the hedonic games literature. Our computational analysis also led to constructive arguments for the existence of a CNS matching for each roommate game, and an IS matching for each marriage game.

There are some interesting contrasts in our results. Although IS and CNS are defined in a symmetric way, we saw that a CNS matching is guaranteed to exist for roommate games whereas even checking the existence of an IS matching is NP-complete. 
In the case of marriage games, we noticed a contrast between the complexity of computing an IS matching and computing a NS matching. Finally, one may naively expect that group-based stability like core may be harder to examine than individual-based stability. However, we note that for marriage games, finding Nash stable matching is intractable whereas computing a core stable matching is polynomial-time solvable.

It will be interesting to consider individual-based stability in other models such as hospitals/residents matching. 
It remains to be seen whether the absence of ties in the preferences can affect the NP-completeness results.
\citet{Knut76a} proved that core deviations can cycle in marriage games. It is also easy to see that Nash deviations can also cycle in marriage games. It will be interesting to investigate  whether IS deviations can cycle in marriage games or not. Finally, the complexity of computing a matching which is both CNS and IR is also open.


\def\bibfont{\small}


\begin{thebibliography}{25}
\providecommand{\natexlab}[1]{#1}
\providecommand{\url}[1]{\texttt{#1}}
\expandafter\ifx\csname urlstyle\endcsname\relax
  \providecommand{\doi}[1]{doi: #1}\else
  \providecommand{\doi}{doi: \begingroup \urlstyle{rm}\Url}\fi

\bibitem[Ackermann et~al.(2011)Ackermann, Goldberg, Mirrokni, R{\"o}glin, and
  V{\"o}cking]{AGM+11a}
H.~Ackermann, P.~W. Goldberg, V.~S. Mirrokni, H.~R{\"o}glin, and
  B.~V{\"o}cking.
\newblock Uncoordinated two-sided matching markets.
\newblock \emph{SIAM Journal on Computing}, 40\penalty0 (1):\penalty0 92--106,
  2011.

\bibitem[Aziz et~al.(2011)Aziz, Brandt, and Seedig]{ABS11b}
H.~Aziz, F.~Brandt, and H.~G. Seedig.
\newblock Stable partitions in additively separable hedonic games.
\newblock In P.~Yolum and K.~Tumer, editors, \emph{Proceedings of the 10th
  International Joint Conference on Autonomous Agents and Multi-Agent Systems
  (AAMAS)}, pages 183--190. IFAAMAS, 2011.

\bibitem[Aziz et~al.(2012)Aziz, Harrenstein, and Pyrga]{AHP12a}
H.~Aziz, P.~Harrenstein, and E.~Pyrga.
\newblock Individual-based stability in hedonic games depending on the best or
  worst players.
\newblock In \emph{Proceedings of the 11th International Joint Conference on
  Autonomous Agents and Multi-Agent Systems (AAMAS)}, 2012.

\bibitem[Ballester(2004)]{Ball04a}
C.~Ballester.
\newblock {NP}-completeness in hedonic games.
\newblock \emph{Games and Economic Behavior}, 49\penalty0 (1):\penalty0 1--30,
  2004.

\bibitem[Bogomolnaia and Jackson(2002)]{BoJa02a}
A.~Bogomolnaia and M.~O. Jackson.
\newblock The stability of hedonic coalition structures.
\newblock \emph{Games and Economic Behavior}, 38\penalty0 (2):\penalty0
  201--230, 2002.

\bibitem[Cechl{\'a}rov{\'a} and Romero-Medina(2001)]{CeRo01a}
K.~Cechl{\'a}rov{\'a} and A.~Romero-Medina.
\newblock Stability in coalition formation games.
\newblock \emph{International Journal of Game Theory}, 29:\penalty0 487--494,
  2001.

\bibitem[Edmonds(1965)]{Edmo65a}
J.~Edmonds.
\newblock Paths, trees and flowers.
\newblock \emph{Canadian Journal of Mathematics}, 17:\penalty0 449--467, 1965.

\bibitem[Gairing and Savani(2010)]{GaSa10a}
M.~Gairing and R.~Savani.
\newblock Computing stable outcomes in hedonic games.
\newblock In S.~Kontogiannis, E.~Koutsoupias, and P.~Spirakis, editors,
  \emph{Proceedings of the 3rd International Symposium on Algorithmic Game
  Theory (SAGT)}, volume 6386 of \emph{Lecture Notes in Computer Science},
  pages 174--185. Springer-Verlag, 2010.

\bibitem[Gairing and Savani(2011)]{GaSa11a}
M.~Gairing and R.~Savani.
\newblock Computing stable outcomes in hedonic games with voting-based
  deviations.
\newblock In \emph{Proceedings of the 10th International Joint Conference on
  Autonomous Agents and Multi-Agent Systems (AAMAS)}, pages 559--566, 2011.

\bibitem[Gale and Shapley(1962)]{GaSh62a}
D.~Gale and L.~S. Shapley.
\newblock College admissions and the stability of marriage.
\newblock \emph{The American Mathematical Monthly}, 69\penalty0 (1):\penalty0
  9--15, 1962.

\bibitem[Gusfield and Irving(1989)]{GuIr89a}
D.~Gusfield and R.~W. Irving.
\newblock \emph{The stable marriage problem: structure and algorithms}.
\newblock MIT Press, Cambridge, MA, USA, 1989.

\bibitem[Hajdukov{\'a}(2006)]{Hajd06a}
J.~Hajdukov{\'a}.
\newblock Coalition formation games: {A} survey.
\newblock \emph{International Game Theory Review}, 8\penalty0 (4):\penalty0
  613--641, 2006.

\bibitem[Horton and Kilakos(1993)]{HoKi93a}
J.~D. Horton and K.~Kilakos.
\newblock Minimum edge dominating sets.
\newblock \emph{SIAM Journal of Discrete Mathematics}, 6\penalty0 (3):\penalty0
  375--387, 1993.

\bibitem[Irving(1985)]{Irvi85a}
R.~W. Irving.
\newblock An efficient algorithm for the ``stable roommates'' problem.
\newblock \emph{Journal of Algorithms}, 6\penalty0 (4):\penalty0 577--595,
  1985.

\bibitem[Irving(1994)]{Irvi94a}
R.~W. Irving.
\newblock Stable marriage and indifference.
\newblock \emph{Discrete Applied Mathematics}, 48\penalty0 (3):\penalty0
  261--272, 1994.

\bibitem[Knuth(1976)]{Knut76a}
D.~E. Knuth.
\newblock \emph{Mariages stables}.
\newblock Les Presses de l'Universit{\'e} de Montr{\'e}al, 1976.

\bibitem[Manlove et~al.(2002)Manlove, Irving, Iwama, Miyazaki, and
  Morita]{MII+02a}
D.~Manlove, R.~W. Irving, K.~Iwama, S.~Miyazaki, and Y.~Morita.
\newblock Hard variants of stable marriage.
\newblock \emph{Theoretical Computer Science}, 276\penalty0 (1--2):\penalty0
  261--279, 2002.

\bibitem[Nisan et~al.(2011)Nisan, Schapira, Valiant, and Zohar]{NSV+11a}
N.~Nisan, M.~Schapira, G.~Valiant, and A.~Zohar.
\newblock Best-response mechanisms.
\newblock In \emph{Proceedings of the 2nd Conference on the Innovations in
  Computer Science (ICS)}, 2011.

\bibitem[Olsen(2009)]{Olse09a}
M.~Olsen.
\newblock Nash stability in additively separable hedonic games and community
  structures.
\newblock \emph{Theory of Computing Systems}, 45:\penalty0 917--925, 2009.

\bibitem[Papai(2007)]{Papa07b}
S.~Papai.
\newblock Individual stability in hedonic coalition formation.
\newblock econ.core.hu/esemeny/korosi/2007/PAPAI07.PDF, 2007.

\bibitem[Ronn(1990)]{Ronn90a}
E.~Ronn.
\newblock {NP}-complete stable matching problems.
\newblock \emph{Journal of Algorithms}, 11:\penalty0 285--304, 1990.

\bibitem[Roth and Sotomayor(1990)]{RoSo90a}
A.~Roth and M.~A.~O. Sotomayor.
\newblock \emph{Two-Sided Matching: {A} Study in Game Theoretic Modelling and
  Analysis}.
\newblock Cambridge University Press, 1990.

\bibitem[Scott(2005)]{Scot05a}
S.~Scott.
\newblock \emph{A Study of Stable Marriage Problems with Ties}.
\newblock PhD thesis, University of Glassgow, 2005.

\bibitem[Sung and Dimitrov(2007)]{Sudi07b}
S.~C. Sung and D.~Dimitrov.
\newblock On myopic stability concepts for hedonic games.
\newblock \emph{Theory and Decision}, 62\penalty0 (1):\penalty0 31--45, 2007.

\bibitem[Sung and Dimitrov(2010)]{SuDi10a}
S.~C. Sung and D.~Dimitrov.
\newblock Computational complexity in additive hedonic games.
\newblock \emph{European Journal of Operational Research}, 203\penalty0
  (3):\penalty0 635--639, 2010.

\end{thebibliography}



\end{document}